\documentclass[runningheads]{llncs}
\usepackage[T1]{fontenc}
\usepackage{graphicx}
\usepackage[disable]{todonotes}
\usepackage{xspace}
\usepackage{proof}
\usepackage{float}
\usepackage{apxproof}
\usepackage{amssymb}
\usepackage{amsmath}
\usepackage{mathtools}
\usepackage{slashed}
\usetikzlibrary{automata}
\usetikzlibrary{arrows.meta}
\usepackage{centernot}
\usepackage{subcaption}
\usepackage{hyperref}
\usepackage{framed}

\newtheoremrep{lemma}{Lemma}[section]
\newtheoremrep{definition}{Definition}[section]
\newtheoremrep{example}{Example}[section]
\newtheoremrep{observation}{Observation}[section]
\newtheoremrep{theorem}{Theorem}[section]
\newtheoremrep{proposition}{Proposition}[section]
\newtheoremrep{corollary}{Corrolary}[section]

\usepackage{turnstile}

\newcommand{\obl}{O}
\newcommand{\frb}{F}
\newcommand{\perm}{P}
\newcommand\bydef[1]{\xLeftrightarrow{\text{#1}}}

\newcommand\cDL{$c\mathcal{DL}$\xspace}
\newcommand\cdl{\cDL}

\newcommand{\Rej}{\textit{Rej}}
\newcommand\mydef{\mathrel{\overset{\makebox[0pt]{\mbox{\normalfont\tiny\sffamily def}}}{=}}}
\newcommand\sat{\models_{s}}
\newcommand\viol{\models_{v}}
\newcommand\qsat{\sttstile{}{}_{s}^p}
\newcommand\qsatv{\sttstile{}{}_{v}^p}
\newcommand\qsatun{\sttstile{}{}_{?}^p}

\newcommand\score{\rho}
\newcommand\guard{\rotatebox[origin=c]{-90}{$\bumpeq$}}
\newcommand\trans{\tau}

\pgfarrowsdeclare{mystealth}{mystealth}
{
	\pgfarrowsleftextend{0pt}
	\pgfarrowsrightextend{.5\pgflinewidth}
}
{
	\pgfpathmoveto{\pgfpoint{0pt}{0.5\pgflinewidth}}
	\pgfpathlineto{\pgfpoint{-\pgflinewidth}{\pgflinewidth}}
	\pgfpathlineto{\pgfpoint{-\pgflinewidth}{-\pgflinewidth}}
	\pgfpathlineto{\pgfpoint{0pt}{-0.5\pgflinewidth}}
	\pgfusepathqfill
}

\newif\ifrevision

\newcommand{\new}[1]{{\ifrevision\color{blue}#1\else#1\fi}}

\begin{document}

\title{Synchronous Agents, Verification, and Blame --- A Deontic View}
\def\orcidID#1{\smash{\href{http://orcid.org/#1}{\protect\raisebox{-1.25pt}{\protect\includegraphics{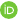}}}}}
\author{Karam Kharraz\inst{1}\orcidID{0000-0002-6908-1756} \and
Shaun Azzopardi \inst{2}\orcidID{0000-0002-2165-3698} \and
Gerardo Schneider \inst{2}\orcidID{0000-0003-0629-6853} \and Martin Leucker\inst{1} \orcidID{0000-0002-3696-9222}}

\authorrunning{Kharraz et al.}
\institute{ISP, University of Lübeck, Germany 
\email{\{kharraz,leucker\}@isp.uni-luebeck.de}\\ \and University of Gothenburg, Sweden
\email{\{shauna,gersch\}@chalmers.se}}

\maketitle

\begin{abstract}
A question we can ask of multi-agent systems is whether the agents' collective interaction satisfies particular goals or specifications, which can be either individual or collective. 
When a collaborative goal is not reached, or a specification is violated, a pertinent question is whether any agent is to blame.
This paper considers a two-agent synchronous setting and a formal language to specify when agents' collaboration is required. We take a \emph{deontic} approach and use \emph{obligations}, \emph{permissions}, and \emph{prohibitions} to capture notions of non-interference between agents. We also handle \emph{reparations}, allowing violations to be corrected or compensated. 
We give trace semantics to our logic, and use it to define blame assignment for violations. We give an automaton construction for the logic, which we use as the base for model checking and blame analysis. We also further provide quantitative semantics that is able to compare different interactions in terms of the required reparations. 
\end{abstract}
\section{Introduction}
Interaction between agents can be adversarial, where each agent pursues its own set of individual goals, or cooperative where the agents collaborate to achieve a collective goal. Verification techniques can help us detect whether such goals may be achieved. 
Agents may also interfere or not cooperate, 
at which point the failure to achieve a goal could be attributed to some agent. In this paper, we develop a \emph{deontic} logic allowing us to specify the anticipated interaction of two agents in the presence of such aspects.

	A deontic logic \cite{McNamara2006,von1951deontic} includes norms as first-class concepts, with \emph{obligations}, \emph{permissions}, and \emph{prohibitions} as basic norms. These concepts are crucial in legal documents and contractual relationships, where the agents are the parties to a contract.\footnote{We use \textit{party} and \textit{agent} interchangeably throughout.} 
	Norms are parameterised by actions/events or propositions and are used to specify what 
	\textit{ought to be}, 
	or the parties 
	\textit{ought to do}.

In this paper, interaction or cooperation of the agents is modelled as the interplay of the individual actions performed by each agent, leading to the concept of cooperative actions. Cooperative actions could be synchronous, i.e., actions at each time point of each agent are meant to describe the possible cooperation, or asynchronous, meaning that actions for cooperation may happen a different time points.\footnote{Observe similarities with synchronous and asynchronous communication.}  We choose synchrony as an abstraction to simplify the concept of cooperation and non-interference between parties. We also study only the setting with two rather than many parties. As such,  we are concerned with {\em two-party synchronous systems}, leaving extensions as future work.

     We re-purpose and extend the syntax of a deontic language from literature \cite{DBLP:conf/jurix/AzzopardiGP16,DBLP:conf/jurix/AzzopardiPS18} into a new deontic logic with denotational semantics appropriate for this two-party setting. Our semantics depends on two notions of {\em informative} satisfaction or violation, which talk about the exact point in time a contract is satisfied or violated. Other features of the logic include the ability to make contracts trigger on matching a regular language, requiring the satisfaction of a contract while one is still within the prefix language of a regular language, and a recursion operator to allow the definition of persistent contracts and repetition.
	
	We extend the semantics with a notion of {\em blame assignment}, to identify which party is responsible for a certain violation. We further use this to define quantitative semantics that counts the number of violations caused by a certain party, which can be used to compare different traces or behaviour of a party. 
	
	We give an exponential automata construction for the logic, transforming a contract specification into an automaton capable of identifying satisfaction, and violation as specified in our semantics. We also provide a model checking algorithm, which is quadratic in the size of the contract automaton, hence exponential in the size of the contract. We re-use this construction for blame analysis, but leave analysis for the quantitative semantics for future work.
	
	The paper organisation follows. Section~\ref{sec:prelim} lays out preliminaries, Section~\ref{sec:logic} presents our logic, and Section~\ref{sec:analysis} presents algorithms for model checking and blame analysis  through automata constructions. Related work is considered in Section~\ref{sec:relatedwork}, and we conclude in Section~\ref{sec:conclusion}.

\vspace{-5pt}
\section{Preliminaries}\label{sec:prelim}
\vspace{-5pt}
We write $\mathbb{N}_\infty$ for $\mathbb{N} \cup \{\infty\}$. Given a finite alphabet $\Sigma$, we write $\Sigma_0$, and $\Sigma_1$ for re-labellings of $\Sigma$ with party identifiers $0$ and $1$, and $\Sigma_{0,1}$ for $\Sigma_0 \cup \Sigma_1$. We use $P[x/y]$ to refer to the syntactic replacement of $x$ in $P$ with $y$, where $P$ can be an automaton ($x$ and $y$ are states), or a specification ($x$ and $y$ are syntactic objects in the language). We write $(*,s)$ to refer to all state pairs with $s$ in the second position, and similarly for $(s,*)$. \\
{\bf Traces} 
For $i \in \mathbb{N}$, $j \in \mathbb{N}_\infty$, and an infinite trace $w$ over sets of actions from a finite alphabet $\Sigma$, we denote the {trace between positions} $i$ and $j$ by $w[i..j]$, including the values at both positions. If $j < i$ then $w[i..j]$ is the empty trace. When $j = \infty$ then $w[i..j]$ is the suffix of $w$ from $i$. 
We write $w[i]$ for $w[i...i]$, and $w \cdot w'$ for concatenation of $w$ and $w'$, which is only defined for a finite word $w$.

Given two traces $w, w'$ over $2^\Sigma$, we define stepwise intersection: $(w \sqcap w')[i] \mydef w[i] \cap w'[i]$, union $(w \sqcup w')[i] \mydef w[i] \cup w'[i]$, and union with party labelling: $(w \sqcup^0_1 w')[i] \mydef w[i] \cup^0_1 w'[i]$, where $E \cup^0_1 E' \mydef \{a_0 \mid a \in E\} \cup \{a_1 \mid a \in E'\}$, i.e. the left actions are labeled by $0$ and the right actions by $1$.
This gives a trace in $\Sigma_{0,1}$.
For instance, given $w=\langle \{a\},\{b\},\{c,d\}\rangle$ and $w'=\langle \{a\},\{e\},\{d,e\}\rangle $, we have that $w[2] \cap w'[2] = \{c,d\}\cap \{d,e\}= \{d\}$ and $w[2] \sqcup^0_1 w'[2] = \{c,d\}\sqcup^0_1\{d,e\}= \{c_0,d_0,d_1,e_1\}$.

Given two traces $w_0$ and $w_1$, over $2^{\Sigma}$, we write $\Vec{w}_i^j$ for the pair $(w_0[i..j], w_1[i..j])$. 
$\Vec{w}_i^j$ is said to be an \textit{interaction}, and when $j \in \mathbb{N}$ a \textit{finite interaction}. Sometimes we abuse notation and treat $\Vec{w}_i^j$ as a trace in $\Sigma_{0,1}$, since it can be projected into such a trace through $\sqcup^0_1$.\\

\textbf{Automata} 
A tuple $A = \langle \Sigma, Q, q_0, \Rej, \rightarrow \rangle$ is an {\em automaton},
where $\Sigma$ is a finite alphabet, $S$ is a finite set of states, $s_0 \in S$ is the initial state, $\Rej \subseteq S$ is a set of rejecting states, and $\rightarrow \in S \times 2^\Sigma \rightarrow (2^S \setminus \emptyset)$ is the transition function ($\rightarrow \in S \times 2^\Sigma \rightarrow S$ when the automaton is deterministic). The language $L(A)$ of automaton $A$ is the set of infinite traces with no prefix reaching a rejecting state. The rejecting language $RL(A)$ of automaton $A$ is the set of infinite traces with a prefix reaching a rejecting state. We write $RL_s(A)$ for the rejecting language through a specific rejecting state $s \in \Rej$.

The \textit{synchronous product} of automata $A$ and $B$ over the same alphabet $\Sigma$, denoted by $A \| B$, is the automaton: $(\Sigma, S_A \times S_B, (s_{0_A}, s_{0_B}), (Rej_A \times S_B) \cup (S_A \times Rej_B), \rightarrow)$ where $\rightarrow$ is the minimal relation such that: for any $E \subseteq \Sigma$, if $s_1 \xrightarrow{E}_A s'_1 $ and $s_2 \xrightarrow{E}_B s'_2$ then  $(s_1, s_2) \xrightarrow{E} (s'_1, s'_2)$.

The \textit{relaxed synchronous product} of automata $A$ and $B$ over the same alphabet $\Sigma$, denoted by $A \|^{r} B$ includes $A \| B$ but allows moving independently when there is no match: if $s_1 \xrightarrow{E}_A s'_1 $ and $\nexists s'_2 \cdot s_2 \xrightarrow{E}_B s'_2$, then  $(s_1, s_2) \xrightarrow{E} (s'_1, s_2)$; and symmetrically.

\textbf{Moore Machines} 
A Moore machine is a 5-tuple $M = (S, s_0, \Sigma_I, \Sigma_O, \delta, \lambda)$ where $S$ is a finite set of states, $s_0 \in S$ is the initial state, $\Sigma_I$ and $\Sigma_O$ are respectively the finite set of input and output actions, $\delta : S \times 2^{\Sigma_I} \rightarrow 2^S$ is a transition function that maps each state and inputs to a next state, and $\lambda: S \rightarrow 2^{\Sigma_O}$ is an output function that maps each state to a set of outputs.

The \textit{product} of a Moore machine $M_1$ over input alphabet $\Sigma_I$ and output alphabet $\Sigma_O$, and Moore machine $M_2$ with flipped input and output alphabets is the automaton: $M_1 \otimes M_2 \mydef (\Sigma_I \cup \Sigma_O, S_1 \times S_2, (s_{0_1}, s_{0_2}), \emptyset, \rightarrow)$ where $\rightarrow$ is the minimal relation such that: for any states $s_1 \in S_1$ and $s_2 \in S_2$, where $s_1 \xrightarrow{\lambda_2(s_2)} s'_1 $ and $s_2 \xrightarrow{\lambda_1(s_1)} s'_2$ then  $(s_1, s_2) \xrightarrow{\lambda_1(s_1) \cup \lambda_2(s_2)} (s'_1, s'_2)$.

{\bf Regular Expressions} 
We use standard syntax for regular expressions. We treat as atomic boolean combinations of actions from $\Sigma_{0, 1}$.\todo{the previous text made it seem we have boolean combinations of regular expressions, as noted by reviewer 1, not it is clearer} The operators are standard: choice, $re + re$ (match either); sequence, $re;re$ (match the first then the second) and the Kleene plus, $re^+$ (match a non-zero finite amount of times in sequence). The language of a regular expression $re$ is a set of finite traces: $L(re) \subseteq (2^{\Sigma_{0,1}})^*$. We abuse notation and write $\Vec{w}_i^j \in L(re)$ for $w_0[i...j] \sqcup_1^0 w_1[i...j] \in L(re)$.

We restrict attention to the \textit{tight language} of a regular expression, containing matching finite traces that have no matching strict prefix: $TL(re) \mydef \{\Vec{w}_i^j \in L(re) \mid \nexists k : k < j \wedge \Vec{w}_i^k \in L(re)\}$. The {\em prefix closure} of the tight language is the set of finite prefixes of the tight language up to a match: $cl(re) \mydef \{\Vec{w}_i^k \mid \exists j : \Vec{w}_i^j \in TL(re) \wedge i \leq k < j\}$. We define the \textit{complement of the prefix closure} as the set of finite traces that do not tightly match the regular expression but whose maximal strict prefix is in the closure of the expression: $\overline{cl}(re) \mydef \{\Vec{w}_i^j \mid (\Vec{w}_i^{j - 1} \in cl(re) \wedge\Vec{w}_i^j \not\in cl(re) \wedge \Vec{w}_i^j \not\in TL(re))\}$. 

We denote by $A(re, s_0, s_{\checkmark}, s_{\times})$ the deterministic finite automaton corresponding to regular expression $re$, with $s_0$, and $s_{\times}$ respectively as the initial and rejecting states and, $s_{\checkmark}$ as a sink state, s.t. $\forall \Vec{w}_i^j \in TL(re) : s_0 \xRightarrow{\Vec{w}_i^j} s_{\checkmark}$, $\forall \Vec{w}_i^j \in {cl}(re) : \exists s: s_0 \xRightarrow{\Vec{w}_i^j} s \wedge s \neq s_{\checkmark} \wedge s \neq s_\times$, and $\forall \Vec{w}_i^j \in \overline{cl}(re) : s_0 \xRightarrow{\Vec{w}_i^j} s_{\times}$.

\section{A Deontic Logic for Collaboration}\label{sec:logic}

In this section, we present the syntax and semantics of \cDL, a deontic logic able to express the extent to which parties should cooperate and non-interfere.

\begin{definition}[\cdl Syntax]
	\label{def:cdl} A \cDL contract $C$ is given by the following grammar, given an alphabet $\Sigma$, regular expressions $re$, a set of variables $\mathbb{X}$, and party labels $p$ from $\{0,1\}$:
	\begin{align*}
		a \in &~ \Sigma_0 \cup \Sigma_1\\
		N :=\ &~\obl_p(a) \mid \frb_p(a) \mid \perm_p(a)\mid \top \mid \perp\\
		C :=\ &~N  \mid C \wedge C \mid C ; C \mid C \blacktriangleright C \mid \\&~\langle re \rangle C \mid re \guard C \mid \mathbb{X} \mid rec X. C
	\end{align*}
\end{definition}
Our setting is that of two-party systems, with one party indexed with $0$ and the other with $1$. As the basic atoms of the language, we have \emph{norms}. These norms are labeled by the party that is the main subject of the norm, and the action that is normed: $\obl_p(a)$ denotes an obligation for party $p$ to achieve $a$; $\frb_p(a)$ denotes a prohibition for party $p$ from achieving $a$, and $\perm_p(a)$ denotes a permission/right for party $p$ to achieve $a$.

We call \cdl specifications \textit{contracts}. Contracts include norms, the atomic satisfied ($\top$), and the transgressed ($\perp$) contract. Contracts can be \emph{conjuncts} ($\wedge$) and \emph{sequentially composed} ($;$). A contract may repair the violation of another ($C \blacktriangleright C'$ means that $C'$ is the \emph{reparation} applied when $C$ is violated). 

Contracts can be triggered when a regular expression matches tightly ($\langle re \rangle C$). A regular expression can also guard $\guard$ a contract $C$, such that an unrecoverable mismatch with it removes the need to continue complying with $C$ in ($re \guard C$).

We allow \textit{recursive} definitions of contracts ($rec~X . C$), where $X \in \mathbb{X}$, with some restrictions. First, we do not allow a contract to have two recursive sub-contracts using the same variable name. Secondly, we have some syntactic restrictions on the contract $C$ appearing inside of the recursion: $C$'s top-level operator is always a sequence, or a regular expression trigger contract, with $X$ only appearing once and on the right-hand side of a sequence, i.e., the expression must be tail recursive. We also require an additional restriction for recursion with the reparation operator: the reparation has to either not be the last operation before $X$ or the whole recursion should be guarded with $re\guard$, the reason behind it is to avoid the procrastination dilemma \cite{jackson1985semantics}.   For example,   $rec~X.\langle re\rangle((C \blacktriangleright C'); X)$ and $ re \guard (rec~X.C \blacktriangleright X)$ are valid, unlike $rec~X.X$, $rec~X.C;(C' \wedge X)$, $rec~X.\langle re\rangle((C ;X); C')$, and $rec~X.C \blacktriangleright X$. Moreover, a recursion variable $X \in \mathbb{X}$ must always be bound when it appears in a contract. 

In our setting, we want to be able to talk about collaborative actions (actions that require both parties to be achieved successfully) and non-interference between the parties (a party not being allowed to interfere with the other party carrying out a certain action). We model both of these using a notion of synchronicity. We will later represent parties as Moore machines; here we talk just about their traces.

We assume two traces over $2^\Sigma$, one for each party: $w_0$ and $w_1$. A party's trace is a record of which actions were enabled (or attempted) by that party. The step-wise intersection of these traces, $w_0 \sqcap w_1$, is the trace of \textit{successful} actions.
Restricting attention to the successful actions misses information about attempts that were not successful. Instead, we give semantics over pairs of party traces, an \textit{interaction}, rather than over $w_0 \sqcap w_1$, allowing us to localise interference.
This setting allows us to model both collaboration and non-interference between the parties in the same way. If the parties are required to collaborate on an action, then they must both propose it  (\textit{obligation}). If instead, the parties should ensure an action is not successful, then at least one of them must not enable it (\textit{prohibition}). If a party is required to not interfere with another party's action, then they must also enable it (\textit{permission}). We refer to actions of one party variously as \emph{proposed}, \emph{attempted}, or \emph{enabled} by that party. We consider an example specification in our language.
\begin{example}\label{ex:robots}
	Consider two possibly distinct robots, 0 and 1, working on a factory floor, with their main goal being to cooperate in placing incoming packages on shelves. Each robot has sensors to identify when a new package is in the queue (\textit{detectProd}), and they must lift the package together (\textit{lift}), and place it on a shelf (\textit{putOnShelf}). Between iterations of this process, the robots are individually allowed to go to their charging ports (\textit{charge0} or \textit{charge1}). If a robot does not help in lifting, it is given another chance:
	\begin{align*}
		\textit{permitCharge} \mydef&~ P_0(\textit{charge0}) \wedge P_1(\textit{charge1})\\
		\textit{lift}(p) \mydef&~ O_p(\textit{lift}) \blacktriangleright O_p(\textit{lift})\\
		\textit{detect\&Lift}(p) \mydef&~ \langle \textit{detectProd}_p \rangle \textit{lift}(p)\\
		\textit{detect\&Place} \mydef&~ (\textit{detect\&Lift}(0) \wedge  \textit{detect\&Lift}(1))\, ; \\&\, (O_{0}(\textit{putOnShelf}) \wedge O_{1}(\textit{putOnShelf}))\\
		\textit{collabRobot} \mydef&~ rec~X. \textit{permitCharge};\textit{detect\&Place};X.
	\end{align*}
\end{example}
\subsection{Informative Semantics}

\begin{figure}[t]
	\[\begin{array}{l@{\quad}c@{\quad}l}
		\Vec{w}_i^j \sat \top 
		&\mydef&~\ i = j \\
		\Vec{w}_i^j \sat\ \perp 
		&\mydef& ~\  \textit{false} \\
		\Vec{w}_i^j \sat \obl_p(a) 
		&\mydef &~\ i = j \wedge a \in w_0[i] \text{ and } a \in w_1[i]
		\\
		\Vec{w}_i^j \sat \frb_p(a) 
		&\mydef&~\ i = j \wedge a \not\in w_p[i] \text{ or } a \not\in w_{1-p}[i]\\
		\Vec{w}_i^j \sat \perm_p(a) 
		&\mydef&~\ i = j \wedge a \in w_p[i] \text{ implies } a \in w_{1-p}[i]) 
		\\
		\Vec{w}_i^j \viol N
		&	\mydef &~\ i = j \wedge \Vec{w}_i^j \not\sat N  \\\\
		\Vec{w}_i^j \sat \langle re \rangle C 
		& \mydef&~\ \Vec{w}_i^j \in \overline{cl}(re) \text{ or } (\exists k < j: \Vec{w}_i^k \in TL(re) \text{ and } \Vec{w}_{k+1}^j \sat C)  \\
		\Vec{w}_i^j \viol \langle re \rangle C 
		&	\mydef &~\ \exists k < j: \Vec{w}_i^k \in TL(re)\text{ and } \Vec{w}_{k+1}^j \viol C\\
		
		\Vec{w}_i^j \sat re\guard C 
		& \mydef&~\ (\Vec{w}_i^j \in \overline{cl}(re) \cup TL(re)    \text{ and } \nexists k<j :  (\Vec{w}_i^k  \viol C )) \\ 
		&  &~\ \text{ or } (\Vec{w}_i^j\in cl(re)   \text{ and }\Vec{w}_i^j \sat C)  \\ 
		\Vec{w}_i^j \viol re \guard C 
		&	\mydef &~\ \Vec{w}_i^j \in cl(re) \text{ and } \Vec{w}_i^j \viol C\\ 
		\Vec{w}_i^j \sat C \wedge C' 
		& \mydef &~\    \Vec{w}_i^k \sat C \text{ and } \Vec{w}_i^l \sat C'  \text{ and } j = max(k,l)\\
		\Vec{w}_i^j \viol C \wedge C' & \mydef &~\ (\Vec{w}_i^j \viol C  \text{ or } \Vec{w}_i^j \viol C') \text{ and } \nexists  k < j : \Vec{w}_i^{k} \viol C \wedge C' \\
		\Vec{w}_i^j \sat C ; C' 
		& \mydef &~\  \exists k < j: \Vec{w}_i^k \sat C \text{ and }\Vec{w}_{k+1}^j \sat C' 
		\\
		\Vec{w}_i^j \viol C ; C' 
		&\mydef &~\ (\exists k < j:  \Vec{w}_i^k \sat C \text{ and } \Vec{w}_{k+1}^j \viol C')  \text{ or } \Vec{w}_i^j \viol C\\
		\Vec{w}_i^j \sat C \blacktriangleright C' 
		& \mydef &~\ \Vec{w}_i^j \sat C \text{ or } (\exists k < j: \Vec{w}_i^k \viol C \text{ and } \Vec{w}_{k+1}^j \sat C') \\
		\Vec{w}_i^j \viol C \blacktriangleright C' 
		& \mydef &~\ \exists k < j : \Vec{w}_i^k \viol C  \text{ and } \Vec{w}_{k+1}^j \viol C'  \\
		\Vec{w}_i^j \sat rec~X.C 
		&\mydef&~\ \Vec{w}_i^j \sat C[X \backslash rec~X . C]  \\	 
		\Vec{w}_i^j \viol rec~X.C 
		&\mydef&~\ \Vec{w}_i^j \viol C[X \backslash rec~X . C]\\\\
		\Vec{w}_i^j \vDash_? C &\mydef&~\ \nexists k \leq j: \Vec{w}_i^k \sat C \text{ or } \Vec{w}_i^k \viol C
	\end{array}\]
	\caption{Informative semantics rules over a finite interaction $\Vec{w}_i^j$.}
	\label{fig:informative}
\end{figure}

The semantics of our language is defined on an \textit{interaction}, i.e. a pair of traces $w_0$ and $w_1$, restricting our view to a slice with a minimal position $i$ and maximal one $j$. For the remainder of this paper, we will refer to this interaction with $\Vec{w}_i^j$. 

In Figure \ref{fig:informative}, we introduce the semantic relations for \textit{informative} satisfaction ($\sat$) and violation ($\viol$). These capture the moment of satisfaction and violation of a contract in a finite interaction. We use this to later define when an infinite interaction models a contract. In Figure \ref{fig:informative} we also capture with $\vDash_?$, when the interaction slice neither informatively satisfies nor violates the contract.

We give some intuition and mention interesting features of the semantics. Note how we only allow the status of atomic contracts to be informatively decided in one time-step (when $i = j$), given they only talk about one action. When it comes to the trigger contract, our goal is to confirm its fulfillment only when we no longer closely align with the specified trigger language. Alternatively, we consider it satisfied if we've matched it previously and subsequently maintained compliance with the contract. Conversely, we would classify a violation if we achieved a close match but then deviated from the contract's terms.
	Regarding the regular expression guard, we have two scenarios for evaluating satisfaction. First, we ensure satisfaction when either we have precisely matched the language or have taken actions preventing any future matching of the guard, with no prior violations or the guarded contract. Second, we verify satisfaction when there's still a possibility of a precise match of the guard, and the guarded contract has already been satisfied. In contrast, a violation occurs when there remains a chance for a precise match in the future of the guard, and a violation of the sub-contract occurs.
	
	The definitions for conjunction and sequence are relatively simple. Note that for conjunction we take the maximum index at which both contracts have been satisfied. Sequence and reparation are similar, except in reparation we only continue in the second contract if the first is violated, while we violate it if both contracts end up being violated. For recursion, we simply re-write variable $X$ as needed to determine satisfaction or violation.

\begin{example}\label{example:2}
	Note how the semantics ensure that, given traces $w_0$ and $w_1$ such that $w_0[0] = w_1[1] = \{\textit{charge0}, \textit{charge1}\}$ then $\Vec{w}_0^0 \sat permitCharge$, i.e. both robots try to charge and allow each other to charge. But if further $w_0[1..3] = \langle \{\textit{detectProd}\},\{\textit{lift}\}, \{\textit{lift}\}\rangle$ and $w_1[1..3] = \langle \{\},\{\}. \{\} \rangle$, then $\Vec{w}_0^3 \viol CollabRobot$, since robot 0 attempted a lift but robot 1 declined helping in lifting.
\end{example}

\begin{toappendix}
	\begin{lemma}\label{lem:recfree}
		Given an interaction, we can always transform a contract with top-level recursion into an equivalent one that is free of recursion.
	\end{lemma}
	\begin{proof}
		We recall that recursion here is restricted, such that $X$ never appears at the top-level of the contract, but always on the right-hand side of a contract that requires another contract or regular expression to be evaluated before. Thus, before the meaning of $X$ needs to be resolved there is always a time-step before where another part of the contract is being evaluated. Thus, given a finite interaction $\Vec{w}_i^j$, which has the size $j - i + 1$, we can always unroll all the recursion in a contract $j - i$ times. By replacing the remaining recursive contracts with $\top$, we get a contract without recursion that is equivalent to the original one for all finite traces of size $j - i$ for informative satisfaction and violation.  \qed
	\end{proof}
\end{toappendix}

Then, we show that if a contract is informatively satisfied (violated), then any suffix or prefix of the interaction cannot also be informatively satisfied (violated):
\begin{lemmarep}[Unique satisfaction and violation]\label{lem:uniquesatunsat}
	If there exists $j$ and $k$ such that $\Vec{w}_i^j \sat C$ and  $\Vec{w}_i^k \sat C$, then $j=k$. Similarly, if there exists $j$ and $k$ such that $\Vec{w}_i^j \viol C$  and $\Vec{w}_i^k \viol C$ then $j = k$.
\end{lemmarep}
\begin{appendixproof}
	We prove this by structural induction (here we use \textit{satisfies}/\textit{violates} for \textit{informatively satisfies}/\textit{violates}):
	\paragraph{Base Case}: For $\top$, $\perp$, $\obl_p$, $\perm_p$, and $\frb_p$ this holds immediately by definition.
	
	\paragraph{Inductive Hypothesis (IH)}: We assume the lemma holds for contracts $C$ and $C'$, for any $i$, $j$, and $k$.
	\paragraph{Inductive Step}: We prove the result holds for contracts built over $C$ and $C'$:
	
	\begin{itemize}
		\item  $\langle re \rangle C$:
		\begin{enumerate}
			\item By definition of $\sat$, we have two different cases when a formula is informatively satisfied:
			\begin{enumerate}
				\item Either there is a point in time ($m$) at which the regular expression tightly matches (which by definition of $TL$ we are assured happens at only one point in a trace), and that $C$ holds on the rest of trace; applying the IH here (setting $i$ to $m$) we are assured the lemma holds for $\langle re \rangle C$.
				\item Or there is no such point in time, but instead $j$ and $k$ are points in time where the trace breaks the safety of the regular expression. The definition here requires that every strict prefix of the trace is in $cl(re) \setminus TL(re)$ and that the full (finite) trace is in $\overline{cl}(re)$. The first condition ensures no prefix is in $\overline{cl}(re)$, by their respective definitions. This ensures that once a $j$ is found where this case holds, there is no bigger (or smaller) $k$, ensuring the lemma.
			\end{enumerate}      
			\item By definition of $\viol$, we have one case, where is a point in time where the regular expression tightly matches, and the contract is violated. Since the point of tight matching $m$ is unique, and by applying the IH on $C$ (setting $i$ to $m$), then the lemma follows.
		\end{enumerate}
		\par
		\item $ re \guard C$: 
		\begin{enumerate}
			\item By definition of $\sat$, we have two different cases when a formula is informatively satisfied:            
			\begin{enumerate}
				\item The trace is in ${cl}(re)$ and it satisfies $C$, and thus the result follows by simply applying IH on $C$.
				\item The trace is in $\overline{cl}(re)$ and there is no prefix that violates $C$. The definition of $\overline{cl}(re)$ ensures that any prefix in not in ${cl}(re)$, and also the definition of ${cl}(re)$ ensures that any extension of $\overline{cl}(re)$ is not in ${cl}(re)$. And thus the other case does not hold for strict prefixes and suffixes of a trace that satisfies this case. Moreover, the definition of $\overline{cl}(re)$ moreover ensures no strict prefix or suffix of the trace is also in $\overline{cl}(re)$, ensuring the required result.
			\end{enumerate}
			\item By definition of $\viol$, we have one case, where the trace is in ${cl}(re)$, and the contract is violated, and thus the required result follows by simply applying the IH on $C$.
		\end{enumerate}
		\item $C \wedge C'$: 
		\begin{enumerate}
			\item For satisfaction, the result follows by applying the IH in both $C$ and $C'$, and taking their maximum. Since the point of satisfaction for both is assured to be unique, so is their maximum.
			\item For violation, one of the contracts must violate (and thus apply the IH on that contract). Both contracts may be violated at different times, but the second requirement of the definition ensures that we are taking the minimum of such time points.
		\end{enumerate}  
		\item $C;C'$: 
		\begin{enumerate}
			\item For satisfaction, the result follows easily by first applying the IH on $C$ and then on $C'$ (setting $i$ to the unique point of satisfaction/violation of $C$).
			\item For violation, the result follows similarly, by applying, for the first case, the satisfaction part of the IH on $C$, and the violation part on $C'$, and for the second case the violation of IH on $C$.
		\end{enumerate}
		
		\item $C \blacktriangleright C'$:
		\begin{enumerate}
			\item For satisfaction we have two cases:
			\begin{enumerate}
				\item  $C$ is satisfied, and the result holds by the IH.
				\item  $C$ is violated, which the IH ensures this is at a unique time point; and $C'$ is then satisfied from that time point, which the IH also ensures is at a unique time point.
			\end{enumerate} 
			\item For violation we have one case, $C$ is violated (at a unique time point, by IH), and from that point, $C'$ is also violated at the end of the trace, which is ensured to be a unique time point, by IH.
		\end{enumerate} 
		
		\item $rec~X. C$: We apply \ref{lem:recfree}, and rely on the inductive hypothesis. \qed
	\end{itemize}
\end{appendixproof}

\begin{proofsketch}
	For the atomic contracts, this is clear. By structural induction, the result follows for conjunction, sequence, and reparation. For the trigger operations, the definition of $TL$ ensures the result. For recursion, note how given a finite interaction there is always a finite amount of times the recursion can be unfolded (with an upper bound of $j-i$) so that we can determine satisfaction or violation in finite time. (See the Appendix for a detailed proof.\footnote{All the proofs of lemmas, propositions and theorems of this and next section may be found in the appendix.})
\end{proofsketch}

If an interaction is not informative for satisfaction, it is not necessarily informative for violation, and vice-versa. But we can show that if there is a point of informative satisfaction then there is no point of informative violation.

\begin{lemmarep}[Disjoint satisfaction and violation]\label{lem:satimpliesnotviol}
	Informative satisfaction and violation are disjoint:  there are no $j,k$ s.t. $\Vec{w}_i^j \sat C$ and $\Vec{w}_i^k \viol C$. 
\end{lemmarep}
\begin{appendixproof}
	We prove this by structural induction that for all possible words, it is impossible for there to be prefixes of an interaction that both informatively satisfy and violate the same contract. 
	\paragraph{Base Case}: For Contracts of $\top$, $\perp$, $\obl_p$, $\perm_p$, and $\frb_p$ this holds by definition.
	\paragraph{Inductive Hypothesis (IH)}: We assume the lemma holds for contracts $C$ and $C'$, for any $i$, $j$, and $k$.
	\paragraph{Inductive Step}: We prove the result holds for contracts built over $C$ and $C'$:
	
	Throughout we start by assuming that the contract is satisfied between $i$ and $j$, and violated between $i$ and $k$. Then we consider all the possible three cases, $j$ is less than, bigger than, or equal to $k$. 
	\begin{itemize}
		\item $\langle re \rangle C$ -- For satisfaction, we assume that either: (1) $\Vec{w}_i^j$ is in the complement closure of $re$, or (2) the prefix can be split into two finite traces, where the first tightly matches $re$ and the other informatively satisfies $C$.
		\begin{enumerate}
			\item $j < k$: Given (1), consider that traces in the complement closure cannot have prefixes in the tight language, or be in the tight language. For (2), the prefix in the tight language is unique, thus the result follows by using the inductive hypothesis on the suffix and $C$.
			
			\item $j > k$: Given (1), the contract cannot be violated, since suffixes in the complement closure cannot be in the tight language. For (2), as before, the prefix in the tight language is unique, thus the result follows by using the inductive hypothesis on the suffix and $C$.
			
			\item $j = k$: Consider that a trace cannot be in the complement closure and the tight language at the same time, while the inductive hypothesis discharges the second proof obligation, as for the other cases.
		\end{enumerate}
		\item $re \guard C$ -- For satisfaction, we assume that either: (1) $\Vec{w}_i^j$ is in the complement closure or tight language of $re$ and the contract was not previously violated, or (2) $\Vec{w}_i^j$ is in the closure, and the contract is informatively satisfied.  
		\begin{enumerate}
			\item $j < k$: Given (1), and that no extension of a trace in the complement closure can be in the closure, the result follows. Given (2), the result follows by the inductive hypothesis.
			
			\item $j > k$: The result needed (no previous violation) follows immediately in case (1). For (2), the result follows by the inductive hypothesis.
			
			\item $j = k$: Given (1), violation cannot be the case given a trace cannot both be in the complement closure and the closure. Given (2) the result follows by the inductive hypothesis.
		\end{enumerate}
		\item $C \wedge C'$ -- The result here follows easily from the inductive hypothesis.
		
		\item $C ; C'$ -- Satisfaction requires $\Vec{w}_i^j$ to be split into two, where $\Vec{w}_i^l$ satisfies $C$ and $\Vec{w}_{l + 1}^j$ satisfies $C'$.
		\begin{enumerate}
			\item $j < k$: By the inductive hypothesis, $C$ must be satisfied on a prefix of $\Vec{w}_i^k$, then $C'$ is violated on the remaining suffix. However, this suffix must be a prefix of $\Vec{w}_{l + 1}^j$, thus the result follows the inductive hypothesis.
			\item The proof for the remaining cases mirrors the previous cases' proof.
		\end{enumerate}
		
		\item $C \blacktriangleright C'$ --- The proof follows similarly to that of sequence.
		
		\item $rec X.C$ -- We apply \ref{lem:recfree}, and rely on the inductive hypothesis.
	\end{itemize}
\end{appendixproof}
\begin{proofsketch}
	The proof follows easily by induction on the structure of C. 
\end{proofsketch}

We can then give semantics to infinite interactions.
\begin{definition}[Models]\label{def:infinitesat} For an infinite interaction $\Vec{w}_0^\infty$, and a \cdl contract $C$, we say $\Vec{w}_0^\infty$ models a contract $C$, denoted by $\Vec{w}_0^\infty \vDash C$, when there is no prefix of the interaction that informatively violates $C$:
	$\Vec{w}_0^\infty \vDash C \mydef \nexists k \in \mathbb{N} \cdot \Vec{w}_0^k \viol C$.
\end{definition}

\subsection{Blame Assignment}
\begin{figure}[t]
	\[\begin{array}{l@{\quad}c@{\quad}l}
		\Vec{w}_i^j \viol^p \top & \mydef& \textit{false}  \\
		\Vec{w}_i^j \viol^p \bot &\mydef& \textit{false}\\
		\Vec{w}_i^j \viol^p \obl_{1-p}(a) & \mydef & i = j \wedge  a \in w_{1-p}[i] \text{ and } a \not\in w_{p}[i] 
		\\ \Vec{w}_i^j \viol^p \obl_p(a) & \mydef & i = j \wedge  a \not\in w_{p}[i]\\
		\Vec{w}_i^j \viol^p \frb_{1-p}(a) &\mydef& \textit{false}
		\\
		\Vec{w}_i^j \viol^p \frb_p(a) &\mydef& i = j \wedge a \in w_p[i] \text{ and } a \in w_{1-p}[i]\\
		\Vec{w}_i^j \viol^p \perm_{1-p}(a) & \mydef & i = j \wedge a \in w_{1-p}[i] \text{ and } a \not\in w_p[i]\\
		\Vec{w}_i^j \viol^p \perm_p(a) &\mydef& \textit{false}\\
		\Vec{w}_i^j \viol^p \langle re \rangle C & \mydef & \exists k < j : \Vec{w}_i^k \sat TL(re) \text{ and } \Vec{w}_{k+1}^j \viol^p C 
		\\
		\Vec{w}_i^j \viol^p re \guard C 
		&\mydef& \Vec{w}_i^j \in cl(re) \text{ and } \Vec{w}_i^j \viol^p C 
		\\
		\Vec{w}_i^j \viol^p C \wedge C' &\mydef& (\Vec{w}_i^j \viol^p C 
		\text{  or  } \Vec{w}_i^j \viol^p C') \text{ and }  \\
		& & (\nexists k < j : \Vec{w}_i^k \viol^{1-p} C \wedge C') \text{ and } \\ & &\lnot(\textit{conflict}(C,C',\Vec{w}_i^{j-1}))\\
		\Vec{w}_i^j \viol^p C \blacktriangleright C' &\mydef & \exists k : \Vec{w}_i^k \viol C \text{ and } \Vec{w}_{k+1}^j \viol^p C'\\
		\Vec{w}_i^j \viol^p rec~X . C & \mydef & \Vec{w}_i^j \viol^p C[X \backslash rec~X . C]
	\end{array}\]
	\caption{Blame semantics rules over a finite interaction $\Vec{w}_i^j$}
	\label{fig:blame}
\end{figure}

We are not interested only in whether a contract is satisfied or violated, but 
also on
\emph{causation} and \emph{responsibility} \cite{ChocklerH03,halpern2015,chockler2004responsibility}. Here we give a relation that identifies when a party is responsible for a violation at a certain point in an interaction. Blame assignment could be specified following multiple criteria, we assign blame when an agent neglects to perform an action it is obliged to do or that another agent is obliged to do (passive blame), or for attempting to do an action it is forbidden from doing (active blame). The blame is forward looking where we identify the earliest cause of violation. Furthermore, we are only interested in causation and not on more advanced features such as "moral responsibility" or "intentionality". The blame semantics is only defined as a violation by party $p$ relation as in $\viol^p$. This semantics is defined in Figure~\ref{fig:blame}.

For blame assignment, the labeling of norms with parties is crucial. 
Here we give meaning to these labels in terms of who is the main subject of the norm in question. For example, consider that $O_0(a)$ can be violated in three ways: either (i) both parties do not attempt $a$, (ii) party $0$ does not attempt $a$ but party $1$ does, or (iii) party $0$ attempts $a$ but party $1$ does not. Our interpretation is that since party $0$ is the main subject of the obligation, party $0$ is blamed when it does not attempt $a$ (cases (i) and (ii)), but party $1$ is blamed when it does not attempt $a$ (case (iii)). The intuition is that by not attempting $a$, party $0$ violated the contract, thus relieving party $1$ of any obligation to cooperate or non-interfere (given party $0$ knows there is no hope for the norm to be satisfied if they do not attempt $a$). We use similar interpretations for the other norms.

Another crucial observation is that violations of a contract are not necessarily caused by a party. For example, the violated contract $\perp$ cannot be satisfied. Moreover, norms can conflict, e.g., $O_p(a) \wedge F_p(a)$. Conflicts are not immediately obvious without some analysis, e.g., $\langle re\rangle O_p(a) \wedge \langle re'\rangle F_p(a)$ (where there is some interaction for which $re$ and $re'$ tightly match at the same time). We provide machinery to talk about conflicts, to avoid unsound blaming, by characterising two contracts to be conflicting when there is no way to satisfy them together.

\begin{definition}[Conflicts]
	Two contracts $C$ and $C'$ are in conflict after a finite interaction $\Vec{w}_i^j$ if at that point their conjunction has not been informatively satisfied or violated yet, but all possible further steps lead to its violation:
		$\textit{conflict}(C,C', \Vec{w}_i^j) \mydef  \nexists \Vec{w}' : \Vec{w'}_i^j = \Vec{w}_i^j \wedge \Vec{w'}_i^{j + 1} \not\viol C \wedge C'$.
\end{definition}

Another instance of a conflict can be observed between $C_1= \obl_{0}(a); \frb_{1}(c)$ and $C_2= \obl_{0}(b) \blacktriangleright \obl_{0}(c)$ at the second position. This can be demonstrated with a trace of length one, $\langle a_0 ; a_1 \rangle$, where the obligation to achieve $c$ for party $0$ and the prohibition to achieve $c$ for party $1$ have to be enforced simultaneously.

\begin{example} Recall the violating example in Example.~\ref{example:2}, where robot 1 declines in helping lifing, twice. Clearly in that case $\Vec{w}_0^3 \viol^1 \text{ collabRobot }$. However, if robot 0 did not attempt a \textit{lift} in position 3 (i.e., to attempt to satisfy the reparation), the blame would be on the other agent.
\end{example}

From the definition of blame it easily follows that a party is blamed for a violation only when there is a violation:

\begin{proposition}\label{blameviol}
	If a party $p$ is blamed for the violation of $C$ then $C$ has been violated:
	$\exists p \cdot \Vec{w}_i^j \viol^p C$ implies $\Vec{w}_i^j \viol C$.
\end{proposition}
\begin{proof}
	Note how each case of $\viol^p$ implies its counterpart in $\viol$.
\end{proof}

But the opposite is not true:

\begin{proposition}\label{violimpblame} 
	A contract may be violated but both parties be blameless:
	$\Vec{w}_i^j \viol C$  does not imply $\exists p \cdot \Vec{w}_i^j \viol^p C$.
\end{proposition}
\begin{proof}
	Consider their definitions on $\bot$, and given conjunction and the presence of conflicts.
\end{proof}

\begin{proposition}[Satisfaction implies no blame]\label{satimpnoblame} 
	Satisfaction of contract $C$ means that no party will get blamed:
	$\Vec{w}_i^j \sat C $ implies $\nexists p \cdot \Vec{w}_i^j \viol^p C$
\end{proposition}
\begin{proof}
	Assume the contrary, i.e. that $C$ is satisfied but party $p$ is blamed. By Proposition~\ref{blameviol} then there is a violation, but Lemma~\ref{lem:satimpliesnotviol} implies we cannot both have a satisfaction or violation.
\end{proof}

\begin{observation}%
	For any contract $C^{\slashed{\bot}}$ defined on \cdl free of $\bot$ and free of conflicts, the violation of a contract $C^{\slashed{\bot}}$ leads to blame. 
\end{observation}

\begin{observation}[Double blame]
	Double blame in \cdl for both parties $p$ and $1-p$ is possible. Consider $C= \obl_p(a) \wedge \obl_p(b)$. Violation of the left-hand side by $p$ and the violation of the right-hand side by $1-p$ can happen at the same time.
\end{observation}

\vspace{-5pt}
\subsection{Quantitative Semantics}

\begin{figure}[htbp]\new{
		\[\begin{array}{l@{\quad}r@{\quad}l}
			\Vec{w}_i^j,\score \qsat N  & \mydef & \score = 0 \wedge \Vec{w}_i^j \sat N \\
			\Vec{w}_i^j,\score \qsatv N  &\mydef & 
			\left\{
			\begin{array}{lll}
				\score = 1 &\quad \text{if} & \Vec{w}_i^i \viol^{p} N \\
				\score = 0 &\quad \text{if} & N = \perp \vee \Vec{w}_i^i \viol^{1-p} N\\
			\end{array}
			\right.\\
			\Vec{w}_i^j,\score \qsatun N  & \mydef & \textit{false}\\
			\Vec{w}_i^j,\score  \qsat \langle re \rangle C &  \mydef &
			\left\{
			\begin{array}{lll}
				\score  =0 &\quad \text{if} & \Vec{w}_i^j \in \overline{cl}(re) \\
				\Vec{w}_{k+1}^j,\score \qsat C &\quad \text{if} &  \exists k:\Vec{w}_i^k \in TL(re) \\
			\end{array}
			\right.\\
			\Vec{w}_i^j, \score \qsatv \langle re \rangle C     &
			\mydef&
			\exists k : \Vec{w}_i^k \in TL(re)	\text{ and } \Vec{w}_{k+1}^j,\score \qsatv C   \\ 
			\Vec{w}_i^j,\score \qsatun \langle re \rangle C     & \mydef&
			\left\{
			\begin{array}{lcl}
				\score=0 & \text{if} & \nexists k \leq j : \Vec{w}_i^k \in   
				TL(re) \\
				\Vec{w}_{k+1}^j,\score \qsatun C &\quad \text{else} & \exists k \leq j : \Vec{w}_i^k \in   
				TL(re)\\
			\end{array}
			\right.\\
			\Vec{w}_i^j,\score \qsat  re \guard C   & \mydef & 
			\left\{
			\begin{array}{lll}
				\Vec{w}_{i}^{j-1},\score \qsatun  C & \text{if} & \Vec{w}_i^j \in \overline{cl}(re) \cup TL(re)\\
				& & \textit{ and } \nexists \score', k < j: \Vec{w}_{i}^{k},\score' \qsatv  C\\
				\Vec{w}_i^j,\score \qsat C & \text{if } &  \Vec{w}_i^j \in {cl}(re) \\
			\end{array}
			\right.\\
			\Vec{w}_i^j,\score \qsatv  re \guard C   & \mydef &      \Vec{w}_i^j \in {cl}(re) \text{ and }   \Vec{w}_i^j,\score \qsatv C\\
			\Vec{w}_i^j,\score \qsatun  re \guard C   & \mydef& \Vec{w}_i^j \in {cl}(re) \text{ and }  \Vec{w}_i^j,\score \qsatun C \\
			\Vec{w}_i^j,(\score_1 +\score_2) \qsat C \wedge C'
			& \mydef &
			\exists k,l : \Vec{w}_i^k,\score_1 \qsat C \text{ and } \Vec{w}_i^l,{\score_2} \qsat C'  \\
			& &\text{ and } j = max(k,l)  \\
			\Vec{w}_i^j,(\score_1 + \score_2) \qsatv C \wedge C'  & \mydef &
			((\Vec{w}_i^j,\score_2 \qsatv C' \text{ or } \Vec{w}_i^j,\score_2 \qsat C')\\
			&&\qquad \text{ and } \Vec{w}_i^j,\score_1 \qsatv C )\\
			&&\text{ or }  ((\Vec{w}_i^j,\score_1 \qsatv C' \text{ or } \Vec{w}_i^j,\score_1 \qsat C')\\
			&&\qquad \text{ and } \Vec{w}_i^j,\score_2 \qsatv C')\\
			&&\text{ or }  (\Vec{w}_i^j,\score_1 \qsatv C  \text{ and } \Vec{w}_i^j,\score_2 \qsatv C'\\
			&&\qquad \text{ and } \neg \textit{conflict}(C, C', \Vec{w}_i^{j-1}))\\
			&&\text{ or }  (\Vec{w}_i^{j-1},\score_1 \qsatun C  \text{ and } \Vec{w}_i^{j-1},\score_2 \qsatun C'\\
			&&\qquad  \text{ and } \textit{conflict}(C, C', \Vec{w}_i^{j-1}))
			\\
			\Vec{w}_i^j,(\score_1 + \score_2) \qsatun C \wedge C' & \mydef & 
			(\Vec{w}_i^j,\score_1 \qsat C  \text{ and } \Vec{w}_i^j,\score_2 \qsatun C')\\
			&& \text{ or }  (\Vec{w}_i^j,\score_1 \qsatun C  \text{ and } \Vec{w}_i^j,\score_2 \qsat C')
			\\ 
			\Vec{w}_i^j,(\score_1 + \score_2) \qsat C;C' &\mydef & \exists k < j: \Vec{w}_i^k,\score_1 \qsat C \text{ and } \Vec{w}_{k+1}^j,{\score_2} \qsat C'\\
			\Vec{w}_i^j,\score \qsatv C ; C'  & \mydef& 
			\left\{
			\begin{array}{lll}
				\score = \score_{1} + \score_{2} & \text{if} & \exists k < j : \Vec{w}_i^k,\score_{1} \qsat C\\
				&& \text{ and } \Vec{w}_{k+1}^j,\score_{2} \qsatv C'\\
				\Vec{w}_i^j,\score \qsatv  &\text{else} & \\
			\end{array} \right. \\
			\Vec{w}_i^j,(\score_1 + \score_2) \qsatun C~;~C'  & \mydef &
			\Vec{w}_i^j,\score \qsatun C \\
			&&\quad\text{ or } (\Vec{w}_i^k,\score_1 \qsat C \text{ and } \Vec{w}_{k+1}^j,{\score_2} \qsatun C') \\
			\Vec{w}_i^j,\score \qsat C \blacktriangleright C'  &\mydef&
			\Vec{w}_i^j,\score \qsat C\\
			&& \text{ or } (\exists k < j : \Vec{w}_i^k,\score_1 \qsatv C\\ 
			&& \quad \text{ and } \Vec{w}_{k+1}^j,\score_2 \qsat C \wedge \score = \score_1 + \score_2)\\
			
			\Vec{w}_i^j,(\score_1 + \score_2) \qsatv C \blacktriangleright C'  & \mydef &  (\exists k < j:   \Vec{w}_i^k,{\score_{1}} \viol^p C \text{ and } 
			\Vec{w}_{k+1}^j,{\score_2} \qsatv C') \\
			\Vec{w}_i^j,\score \qsatun C \blacktriangleright C'  &\mydef&  
			\left\{
			\begin{array}{lll}
				\score = \score_1  + \score_2 & \text{if} &\Vec{w}_i^k,\score_1 \qsatv C\\
				&& \text{ and }\Vec{w}_{k+1}^j,{\score_2} \qsatun C'\\
				\Vec{w}_i^j,\score \qsatun C & \text{else}
			\end{array} \right. \\
			\Vec{w}_i^j,\score \sttstile{}{}_{\gamma}^p rec~X.C &\mydef& \Vec{w}_i^j,\score \sttstile{}{}_{\gamma}^p C[X \setminus rec~X.C] \qquad \text{for } \gamma \in \{s, v, ?\}
		\end{array} 
		\]}
		
	\caption{Quantitative semantics rules over a finite interaction $\Vec{w}_i^j$.}
	\label{fig:quant}
\end{figure}

While it is possible to assign blame to one party for violating a contract, other qualitative metrics can provide additional information about the violation. These metrics can determine the number of violations caused by each party, as well as the level of satisfaction with the contract. To assess responsibility for contract violations, we introduce the notion of a mistake score, $\score$, for each party, enabling us to calculate a \textit{responsibility degree}. It is important to note that our language permits reparations, whereby violations can be corrected in the next time step. However, interactions that are satisfied with reparations are not considered ideal. We present quantitative semantics to compare satisfying interactions based on the number of repaired violations a party incurs. We define relations that track the number of repaired violations attributed to each party with a mistake score, $\score$, written $\qsat$ for informative satisfaction and $\qsatv$ for informative violation of the contract. We can also keep track of the number of violations when the trace is not informative through $\qsatun$. 
Figure~\ref{fig:quant} provides a definition of this semantics. \new{Note this definition intersects the previous semantic definitions, and due to space constraints, we do not re-expand that further. The addition is that we are disambiguating some cases to identify when to add to the score to identify a violation caused by $p$. For example, see the definition of $\qsatv$ for a norm.}

\begin{example}
	Consider again Example~\ref{ex:robots}, and consider the finite interaction
	$(\langle \{\textit{charge0}\}, \{\textit{detectProd}\}, \{\}, \{\textit{lift}\}\rangle $ and $ \langle \{\textit{charge0}\}, \{\}, \{\textit{lift}\}, \{\textit{lift}\}\rangle$. Note how this
	will lead to robot 0 being given a score of one since on the third step there is a violation that is repaired subsequently.
\end{example}

\begin{lemmarep}[Soundness and completeness]
	The quantitative semantics is sound and complete with regard to the informative semantics:
	$\Vec{w}_i^j \models_{\gamma} C \iff \exists \score_1, \score_2 : \Vec{w}_i^j,{\score_1} \sttstile{}{}_{v}^p C \text{ and } \Vec{w}_i^j,{\score_2} \sttstile{}{}_{v}^{1-p}$ with $\gamma \in \{s, v, ?\}.$
\end{lemmarep}
\begin{proofsketch}
	By induction on the quantitative semantics and informative semantics.
\end{proofsketch}
\begin{appendixproof}
	We prove this by structural induction on C.
	\paragraph{Base Case}: For atomic contracts the result follows easily.
	\paragraph{Inductive Hypothesis}: We assume the lemma holds for contracts $C$ and $C'$, for any $i$ and $j$.
	\paragraph{Inductive Step}: We prove the result holds for contracts built over $C$ and $C'$:
	\begin{itemize}
		\item $\langle re \rangle C$ --  
		\begin{enumerate}
			\item For satisfaction, note the conditions in the quantitative semantics cover all the cases in the informative satisfaction semantics, thus by the inductive hypothesis the result immediately follows.
			\item The argument for violation follows similarly.
			\item For $?$, consider that the cases imply either that the finite interaction is inside the closure (thus there is no satisfaction and no violation), or that a prefix is in the tight language, and the status of $C$ is unknown. The result follows by applying the inductive hypothesis on the latter.
		\end{enumerate}
		\item  $re \guard C$ --
		\begin{enumerate}
			\item For satisfaction, note how the first condition ensures we are in the complement closure or tight language, and that the contract has not been previously violated, as in the informative semantics, while applying the inductive hypothesis on the return hypothesis ensures the contract was not satisfied in the previous step in both the quantitative and informative, as required, but gives a value for $\score$. For the second condition, this also matches the respective disjunct in the informative semantics, and we apply the inductive hypothesis on the result to get what is needed.
			\item For violation, again the definition mirrors the informative semantics and we can just apply the inductive hypothesis.
			\item For $?$ we can apply the inductive hypothesis, and discern that the requirement that the finite interaction is in the closure of $re$ follows from the informative definition of $?$.
		\end{enumerate}
		\item $C \wedge C'$ --
		\begin{enumerate}
			\item For satisfaction, the definitions mirror each other, and applying the inductive hypothesis gives us what we need.
			\item For violation, in the quantitative semantics we enumerate every possible way to violate the conjunction. This corresponds to the informative definition, however, we present it differently, by checking that no prefix also violates the contract. The two approaches are equivalent. By applying the inductive hypothesis, and comparing each resulting disjunct from the quantitative definition to the informative definition we get the result we need. Note the conditions cover all possible combinations leading to a violation. For example, for the first case $(\Vec{w}_i^j,\score_1 \qsatv C  \text{ and } (\Vec{w}_i^j,\score_2 \qsatv C' \text{ or } \Vec{w}_i^j,\score_2 \qsat C'))$ iff, by the inductive hypothesis, $(\Vec{w}_i^j \viol C \text{ and } (\Vec{w}_i^j \vDash_? C' \text{ or } \Vec{w}_i^j \sat C')$, which in turn implies the informative definition. 
			
			Here we also deal with conflicts for the purpose of incrementing the score in case of a conflict, this does not affect the correctness of our argument.
			
			\item For $?$, note how applying the inductive hypothesis, the results cover the exact cases where the determination remains unknown.
		\end{enumerate}
		\item $C ; C'$
		\begin{enumerate}
			\item For satisfaction and violation, the quantitative definition mirrors the informative definition, and the result then immediately follows from the inductive hypothesis.
			\item For $?$, the quantitative definition can be easily seen to cover all cases where there is no violation in the informative definition, the result follows from the inductive hypothesis. 
		\end{enumerate}
		\item  $C \blacktriangleright C'$ --
		\begin{enumerate}
			\item For satisfaction and violation, the quantitative definition mirrors the informative definition, and the result then immediately follows from the inductive hypothesis.
			\item For $?$, the quantitative definition can be easily seen to cover all cases where there is no violation in the informative definition, the result follows from the inductive hypothesis. 
		\end{enumerate}
		\item $ rec X. C$
		\begin{enumerate}
			\item By Lemma~\ref{lem:recfree} and the inductive hypothesis the result follows.
		\end{enumerate}
	\end{itemize}
\end{appendixproof}

\begin{lemmarep}[Fairness of the Quantitative semantics]
	\new{The quantitative semantic is fair, meaning that if the score of a player $p$ is $\score$ then $p$ is to be blamed for non-fulfilling $\score$ norms of the contract:
		$\Vec{w}_i^{j},\score \sttstile{}{}_{\gamma}^p C \implies \exists N_1 \dots N_\score \in \textit{subcontracts}(C) : \Vec{w}_{i}^j \sttstile{}{}_{\gamma}^p N_i$ with $\gamma \in \{s, v, ?\}$, where $\textit{subcontracts}(C)$ is a multiset containing the subcontracts of $C$, up to how often they appear.}
\end{lemmarep}
\begin{proofsketch}
	We prove, this by structural induction, noting that the score only increases when $p$ is blamed for the violation of a norm, while the inductive case easily follows from the inductive hypothesis.
\end{proofsketch}
\begin{appendixproof}
	We prove this by structural induction on C.
	\paragraph{Base Case}: By analysis of the definition, one can easily see that $\score$ is set to $1$ only when party $p$ is to blame for the violation of a norm.
	\paragraph{Inductive Hypothesis (IH)}: We assume the lemma holds for contracts $C$ and $C'$, for any $i$, $j$, and $k$.
	\paragraph{Inductive Step}: We prove the result holds for contracts built over $C$ and $C'$:
	\begin{itemize}
		\item In all cases except, $\score$ is not modified but its value depends on the evaluation of sub-contracts. This allows us to dispatch all the proof obligations to the inductive hypothesis.
	\end{itemize}
\end{appendixproof}

\vspace{-5pt}

\section{Analysis}\label{sec:analysis}

In this section, we define an automata-theoretic approach to analyzing \cdl contracts, through a construction to a safety automaton. We use this for model checking and blame analysis, \new{but leave the application for quantitative analysis for future work}.

\subsection{Contracts to Automata}

We give a construction from \cdl contracts to automata that recognize interactions that are informative for satisfaction or violation. For brevity, we keep the definition of the automata symbolic, with transitions tagged by propositions over party actions, representing a set of concrete transitions. The automaton is over the alphabet $\Sigma_{0,1}$ since it requires information about the parties.

\new{
	\begin{definition}\label{def:transm}
		\allowdisplaybreaks
		The deterministic \emph{automaton of contract} $C$ is: 
		\[aut(C) \mydef \langle\Sigma _{0,1}, S, s_0, \{s_B\}, \rightarrow\rangle\text{.}\] We define $\rightarrow$ through the below function $\trans(C, s_0, s_G, s_B, \{\})$ that computes a set of transitions, given
		a contract, an initial state ($s_0$), a state denoting satisfaction ($s_G$), a state denoting violation ($s_B$), and a partial function $V$ from recursion variables ($\mathbb{X}$) to states, characterised by (with $s$ as a fresh state):
		\begin{align*}
			\trans(\top, s_0, s_G, s_B, V) &\mydef \{s_0 \xrightarrow{true} s_G\}\\
			\trans(\perp, s_0, s_G, s_B, V) &\mydef \{s_0 \xrightarrow{true} s_B\}\\
			\trans(\obl_p(a), s_0, s_G, s_B, V) &\mydef \{s_0 \xrightarrow{a_p \wedge a_{1-p}} 
			s_G,  s_0 \xrightarrow{\neg(a_p \wedge a_{1-p})} s_B\}\\
			\trans(\frb_p(a), s_0, s_G, s_B, V)  &\mydef \{s_0 \xrightarrow{\neg(a_p \wedge a_{1-p})} s_G, s_0  \xrightarrow{a_p \wedge a_{1-p}} s_B\}\\
			\trans(\perm_p(a), s_0, s_G, s_B, V)  &\mydef \{s_0 \xrightarrow{a_{p} \implies a_{1-p}} s_G, s_0 \xrightarrow{\neg(a_{p} \implies a_{1-p})} s_B\}\\
			\trans(\langle re \rangle C, s_0, s_G, s_B, V)  &\mydef A(re, s_0, s, s_G) \cup \trans(C, s, s_G, s_B, V) \\
			\trans(re \guard C, s_0, s_G, s_B, V)  &\mydef (A(re, s_0, s_G, s_G) \| \trans(C, s_0, s_G, s_B, V))\\
			&\quad\ [(s_G, *)/s_G][(*, s_B)/s_B][(*,s_G)/s_G]\\
			\trans(C \wedge C', s_0, s_G, s_B, V)  &\mydef (\trans(C,  s_0, s_G, s_B, V) \|^{r} \trans(C', s_0, s_G, s_B, V))\\
			&\quad\ [(s_G,s_{G})/s_G][(s_B, *)/s_B][(*, s_B)/s_B]\\
			\trans(C;C', s_0, s_G, s_B, V)  &\mydef \trans(C, s_0, s, s_B, V) \cup \trans(C', s, s_G, s_B, V)\\
			\trans(C \blacktriangleright C', s_0, s_G, s_B, V)  &\mydef \trans(C, s_0, s_G, s, V) 
			\cup \trans(C', s, s_G, s_B, V) \\
			\trans(X, s_0, s_G, s_B, V)  &\mydef \{s_G \xrightarrow{\epsilon} V(X)\}\\
			\trans(rec X.C, s_0, s_G, s_B, V)  &\mydef \trans(C, s_0, s_G, s_B, V[X \mapsto s_0])
		\end{align*}
		
		We define $\rightarrow'$ as $\trans(C, s_0, s_G, s_B, \{\})$ without all transitions outgoing from $s_G$ and $S_B$, and define $\rightarrow \mydef \rightarrow' \cup \{s_B \xrightarrow{\textit{true}} s_B\} \cup \{s_G \xrightarrow{true} s_G\}$, where $S$ is the set of states used in $\rightarrow$. We assume the $\epsilon$-transitions are removed using standard methods.
\end{definition}}

\new{
	We give some intuition for the construction. The transitions for the atomic contracts follow quite clearly from their semantics. For the trigger contracts, we use a fresh state $s$ to connect the automaton for the regular expression, with that of the contract, ensuring the latter is only entered when the former tightly matches. For the guard contract, we instead synchronously compose ($\|$) both automata (i.e., intersect their languages), getting a set of transitions. Here we also relabel tuples of states to single states. Recall we use $(*,s)$ to match any pair, where the second term is $s$, and similarly for $(s, *)$. Through the sequence of re-labellings, we ensure: first that reaching $s_G$ in the acceptance of the first means; (2) reaching $s_B$ in the second means violation; and (3) if the previous two situations are not the case, reaching $s_G$ in the second means acceptance. 
	
	For conjunction, instead of using the synchronous product, we use the relaxed variant ($\|^r$), since the contracts may require traces of different lengths for satisfaction. This relaxed product allows the `longer' contract to continue after the status of the other is determined. For sequence, we use the fresh state $s$ to move between the automata, once the first contract has been satisfied. For reparation this is similar, except we move between the contracts at the moment the first is violated. For recursion, we simply loop back to the initial state of the recursed contract with an $\epsilon$-transition once the corresponding recursion variable is encountered.
	
}

Note how analyzing states without viable transitions, after applying $\trans$,
can be used for \textit{conflict analysis} of \cdl contracts. For example, when there is a conflict, e.g., $O_p(a) \wedge F_p(a)$, there will be a state with all outgoing transitions to $s_B$.

\begin{theoremrep}[Correctness]\label{thm:autcorrect} An infinite interaction is a model of $C$, iff it never reaches a rejecting state in $aut(C)$: \\$\forall \Vec{w}_0^\infty \cdot \Vec{w}_0^\infty \models C \iff w_0 \sqcup^0_1 w_1 \in L(aut(C))$.
\end{theoremrep}

\begin{proofsketch}
	For the atomic contracts, the correspondence should be clear. By structural induction on the rest: triggering, sequence, and reparation should also be clear from the definition. For conjunction, the relaxed synchronous product makes sure the contract not yet satisfied continues being executed, as required, while the replacements ensure large nestings of conjunctions do not lead to large tuples of accepting or rejecting states. For $\guard$, using the synchronous product ensures the path ends when either is satisfied/violated, as required. 
\end{proofsketch}

\begin{appendixproof}
	First, some observations. Looking at the left-hand side, $w_0 \sqcup^0_1 w_1 \in L(aut(C))$ means that no finite prefix of $w_0 \sqcup^0_1 w_1$ reaches a rejecting state. 
	
	Looking at the right-hand side, $\Vec{w}_0^\infty \models C$ means that there is no finite prefix of $\Vec{w}_0^\infty$ that informatively violates: $\nexists k \in \mathbb{N} \cdot \Vec{w}_0^k \vDash_{\perp} C$. Note how $\Vec{w}_0^k$ has a direct correspondence to $(w_0 \sqcup^0_1 w_1)[0...k]$ (recall $\Vec{w}_0^k$ is $(w_0[0...k], w_1[0...k])$, and $(w_0 \sqcup^0_1 w_1)[0...k] = w_0[0...k] \sqcup^0_1 w_1[0...k]$). We use this direct correspondence implicitly throughout.
	
	We now start the proof, relying on structural induction on $C$.
	\paragraph{Base Case}
	\begin{enumerate}
		
		\item The result should be clear for $\top$ and $\perp$.
		
		\item $\obl_p(a)$:
		\begin{enumerate}
			\item LHS: $\Vec{w}_0^\infty \models \obl_p(a) \bydef{def~\ref{def:infinitesat}} \nexists k. \Vec{w}_0^k \viol \obl_p(a) \bydef{fig.\ref{fig:informative}} \nexists k \cdot 0=k \text{ and } a \notin w_p[0] \text{ or } a \notin w_{1-p}[0]$;
			\item RHS: In $aut(\obl_p(a))$, the only way to reach a rejecting state is to take the transition $s_0 \xrightarrow{\neg(a_p \wedge a_{1-p})} s_B$. Taking this transition implies $\neg(a_p \in w_p[0] \wedge a_{1-p} \in w_{1-p}[0])$, which contradicts the LHS.
		\end{enumerate}
		
		\item $\frb(a)$:
		\begin{enumerate}
			\item LHS: $\Vec{w}_0^\infty \models \frb_p(a) \bydef{def~\ref{def:infinitesat}} \nexists k. \Vec{w}_0^k \viol \frb_p(a) \bydef{} \nexists k \cdot 0=k \text{ and } a \in w_p[0] \text{ or } a \in w_{1-p}[0]$;
			\item RHS: In $aut(\frb_p(a))$, the only way to reach a rejecting state is to take the transition $s_0 \xrightarrow{a_p \wedge a_{1-p}} s_B$. Taking this transition implies $a_p \in w_p[0] \wedge a_{1-p} \in w_{1-p}[0]$ which contradicts the LHS.
		\end{enumerate}
		
		\item $\perm_p(a)$:
		\begin{enumerate}
			\item LHS: $ \Vec{w}_0^\infty \models \perm_p(a) \bydef{def~\ref{def:infinitesat}} \nexists k. \Vec{w}_0^k \viol \perm(a) \bydef{\ref{fig:informative}} \nexists k \cdot 0=k \text{ and } a \in w_p[0] \text{ and } a \notin w_p[0]$;
			\item RHS: In $aut(\perm_p(a)))$ the only way to reach the rejecting state is to take the transition $s_0 \xrightarrow{\neg(a_{p} \implies a_{1-p})} s_B$. Taking this transition implies $\neg(a_p \in w_p[0] \implies a_{1-p} \in w_{1-p}[0])$ which contradicts the LHS.
		\end{enumerate}
	\end{enumerate}
	
	\paragraph{Inductive Hypothesis} We assume the theorem for contracts $C$ and $C'$.
	\paragraph{Inductive Step} We prove it now for contracts built over $C$ and $C'$:
	
	We present the proofs from left-to-right, they also apply from right-to-left, unless explicitly stated.
	\begin{itemize}
		\item $\langle re \rangle C$ --
		\begin{enumerate}
			\item Starting from the LHS, all finite prefixes are of two types:
			\begin{enumerate}
				\item The prefix is not in the tight language of $re$, ensuring that no finite prefix of it reaches $s$ in $A(re, s_0, s, s_G)$, thus the automaton of $C$ is never entered, and $s_B$ is unreachable.
				\item The prefix tightly matches $re$, say $\Vec{w}_0^k$, and its corresponding finite suffix violates $C$, say $\Vec{w}_{k+1}^j$. We can apply the IH on $\Vec{w}_{k + 1}^\infty \vdash C$, concluding that the automaton of $C$ accepts $\Vec{w}_{k + 1}^\infty$. Moreover, by $A(re, s_0, s, s_G)$, $\Vec{w}_0^k$ reaches $s$, the initial state of the automaton for $C$, ensuring $\Vec{w}_0^j$ reaches $s_G$, a sink state, and never $s_B$, such that then we can conclude the RHS.
			\end{enumerate}
		\end{enumerate} 
		
		\item $re \guard C$ --
		\begin{enumerate}
			\item Starting from the LHS, all finite prefixes are of two types:
			\begin{enumerate}
				\item The prefix $\Vec{w}_0^k$ is not in the closure of the language. It, or a prefix of it, thus reaches $s_G$ in $A(re, s_0, s_G, s_G)$. On the RHS, we compose this automaton with the automaton for $C$. Note $s_G$ has no outgoing transitions in $A(re, s_0, s_G, s_G)$, therefore states of the form $(s_G, *)$ have no outgoing transitions. Similarly for states of the form $(*, s_B)$. Given the relabelling, we ensure reaching $s_G$ in $re$'s automaton means acceptance, giving us the result we need.
				\item There is a prefix in the closure of $re$ and this prefix does not violate $C$. Any infinite extension is thus a model of $C$, including $\Vec{w}_0^\infty$, thus this interaction is a model of $C$. By the inductive hypothesis, it is also in the language of $aut(C)$. Moreover, the definition of the synchronous composition ensures this interaction is also accepted by the composed automaton, giving us the result we need.
			\end{enumerate}
		\end{enumerate}
		
		\item $C \wedge C'$ --- 
		\begin{enumerate}
			\item Starting from the LHS:
			\begin{enumerate}
				\item There is no prefix that informatively violates $C$, and no prefix that informatively violates $C'$. By the inductive hypothesis then the interaction must be in the automaton of each. For the conjunction, we computed the relaxed synchronous product of the two automata that computes the union of two languages. While relabelling ensures violating in one automaton means violating in the composed one. This gives us the result needed.
			\end{enumerate} 
		\end{enumerate}
		
		\item $C ; C'$ --
		\begin{enumerate}
			\item Starting from the LHS,  every finite prefix does not informatively violate $C$ (0) and moreover can be of two types:
			\begin{enumerate}
				\item It has no strict prefix that satisfies $C$. Applying the inductive hypothesis (given (0)) ensures any infinite continuation of such a prefix remains in the language of the automaton of $C$.
				\item It has a strict prefix that satisfies $C$ and the corresponding finite suffix does not violate $C'$. Applying the inductive hypothesis, we can conclude that the latter prefix reaches $s$, and then the suffix continues executing in the automaton of $C'$, where it never reaches $s_B$ (given the inductive hypothesis on $C'$).
			\end{enumerate}
		\end{enumerate}
		
		\item $C \blacktriangleright C'$ --
		\begin{enumerate}
			\item Starting from the LHS, every finite prefix is of two types:
			\begin{enumerate}
				\item The prefix does not informatively violate $C$. By applying the inductive hypothesis directly the result follows.
				\item The prefix informatively violates $C$ and it has no corresponding finite suffix that violates $C'$. Applying the inductive hypothesis on both, the prefix will reach the intermediate bad state ($s$) in $C$'s automaton, and thus enter $C'$'s automaton, where it never reaches the bad state $s_B$, as required.
			\end{enumerate}
		\end{enumerate} 
		
		\item $rec X.C$ --
		\begin{enumerate}
			\item The correctness of this depends on the correctness of $C$, note how the automaton loops back to the initial state once a recursion variable is reached.
		\end{enumerate}
		\qed
	\end{itemize}
\end{appendixproof}

\begin{corollary}\label{thm:autbad} An infinite interaction is not a model of $C$, if and only if it reaches a rejecting state in $aut(C)$: $\forall (w_0,w_1) \not\models C \iff \exists j \in \mathbb{N} \cdot s_0 \xRightarrow{(w_0 \sqcup^0_1 w_1)[0...j]} s_B$.
\end{corollary}
\begin{proofsketch}
	Follows from Theorem.~\ref{thm:autcorrect} and completeness (up to rejection) of $aut(C)$.
\end{proofsketch}
\textbf{Complexity} From the translation note that without regular expressions the number of states and transitions is linear in the number of sub-clauses and operators in the contract, but is exponential in the presence of regular expressions.\footnote{For example, a contract $rec X. \top;(O_0(a) \wedge P_1(b));X$ has size 8 (note normed actions are not counted).}

\subsection{Model Checking}

\new{
	We represent the behaviour of each party as a Moore machine ($M_0$, and $M_1$). For party 0, the input alphabet is $\Sigma_1$ and the output alphabet is $\Sigma_0$, and vice-versa for party 1. We characterise the composed behaviour of two parties by defining the product of the two dual Moore machines: $M_0 \otimes M_1$, getting an automaton over $\Sigma_0 \cup \Sigma_1$.
	
	We can then compose this automaton that represents the interactive behaviour of the parties with the contract's automaton, $(M_0 \otimes M_1) \| aut(C)$. Then, if no rejecting state is reachable in this automaton, the composed party's behaviour respects the contract.
	
	\begin{theorem}[Model Checking Soundness and Completeness]\label{thm:mc}
		$\emptyset = RL((M_0 \otimes M_1) \| aut(C))$ iff $\nexists \Vec{w}_0^\infty : w_0 \sqcup^0_1 w_1 \in L(M_0 \otimes M_1) \wedge \Vec{w}_0^\infty \viol C$.
	\end{theorem}
	\begin{proof}
		Consider that $\|$ computes the intersection of the languages, while Theorem.~\ref{thm:autcorrect} states that $L(aut(C))$ contains exactly the traces satisfying $C$ (modulo a simple technical procedure to move between labelled traces and pairs of traces). Then it follows easily that $RL((M_0 \otimes M_1) \| aut(C))$ is empty only when there is no trace in $(M_0 \otimes M_1)$ that leads to a rejecting state in $aut(C)$. The same logic can be taken in the other direction. \qed

\end{proof}}

\subsection{Blame Assignment}

For the blame assignment, we can modify the automaton construction by adding two other violating states: $s^0_B$ and $s^{1}_B$, and adjust the transitions for the basic norms accordingly.

\new{\begin{definition}
		The deterministic \emph{blame automaton of contract} $C$ is: 
		\[blAut(C) \mydef \langle\Sigma _{0,1}, S, s_0, 
		\{s_B, s^0_B, s^1_B, (s^0_B, s^1_B)\}, \rightarrow\rangle\text{}\]
		
		We define $\rightarrow$ through the function $\trans(C, s_0, s_G, s^0_B, s^1_B, V)$ that computes a set of transitions, as in Definition~\ref{def:transm} but now assigning blame by transitioning to the appropriate state. We focus on a subset of the rules, given limited space, where there are substantial changes\footnote{The missing rules essentially mirror the previous construction with the added states, and the different domains.}:

		{\allowdisplaybreaks
			\begin{align*}
				\trans(\obl_p(a), s_0, s_G, s^0_B, s^1_B, V) &\mydef \{s_0 \xrightarrow{a_p \wedge a_{1-p}} s_G, s_0 \xrightarrow{\neg a_p} s^p_B, s_0 \xrightarrow{a_p \wedge \neg a_{1-p}} s^{1-p}_B \}\\
				\trans(\frb_p(a), s_0, s_G, s^0_B, s^1_B, V)  &\mydef \{s_0 \xrightarrow{\neg(a_p \wedge a_{1-p})} s_G, s_0 \xrightarrow{a_p \wedge a_{1-p}} s^{p}_B\}\\
				\trans(\perm_p(a), s_0, s_G, s^0_B, s^1_B, V)  &\mydef \{s_0 \xrightarrow{a_{p} \implies a_{1-p}} s_G, s_0 \xrightarrow{a_{p} \wedge \neg a_{1-p}} s^{1-p}_B\}\\
				\trans(C \blacktriangleright C', s_0, s_G, s^0_B, s^1_B, V)  &\mydef \trans(C, s_0, s_G, s^0, s^1, V) \\
				&\quad \cup \trans(C', s^0, s_G, s^0_B, V)
				\cup \trans(C', s^1, s_G, s^1_B, V) \\
		\end{align*}}
		
		\noindent
		
		Given $\rightarrow' = \trans(C, s_0, s_G, s_B, \{\})$, $\rightarrow$ is defined as $\rightarrow'$ with the following transformations, in order: (1) any tuple of states containing both $s^0_B$ and $s^1_B$ is relabelled as $(s^0_B, s^1_B)$; (2) any tuple of states containing $s^0_B$ ($s^1_B$) is relabelled as $s^0_B$ ($s^1_B$); (3) any state for which all outgoing transitions go to a bad state are redirected to $s_B$; (4) any tuple of states containing $s_G$ is relabelled as $s_G$; and (5) all bad states and $s_G$ become sink states. $S$ is the set of states used in $\rightarrow$. We assume the $\epsilon$-transitions are removed using standard methods.
	\end{definition}
	
	\new{Note how this automata simply refines the bad states of the original automata construction, by assigning blame for the violation of norms through a transition to an appropriate new state. While the post-processing (see (3)), allows violations caused by conflicts to go instead to state $s_B$, where no party is blamed.}
	
	Then we prove correspondence with the blame semantics:
	\begin{theoremrep}[Blame Analysis Soundness and Completeness]
		Where $RL_p$, for $p \in \{0,1\}$, is the rejecting language of the automaton through states that pass through $s^p_B$ or the tuple state $(s^0_B, s^1_B)$:
		
		$\emptyset = RL_p((M_0 \otimes M_1) \| blAut(C))$ iff $\nexists w_0, w_1 \in (2^\Sigma)^* : w_0 \sqcup^0_1 w_1 \in L(M_0 \otimes M_1)  \wedge (w_0,w_1) \viol^p C$.
\end{theoremrep}}

\begin{proof}
	This follows from a slight modification of Corollary.~\ref{thm:autbad} (since here we are just refining the bad states of $aut(C)$) with the replacement of $s_B$ by party-tagged states when someone can be blamed, and from a similar argument to Theorem.~\ref{thm:mc}.
\end{proof}

This automaton can be used for model checking as before, but it can also answer queries about who is to blame.
\begin{example}
	We illustrate in Figure~\ref{fig:composition} an example of two Moore machines representing the behaviour of two parties (Figures~\ref{fig:moore1} and \ref{fig:moore2}). Note these are deterministic, therefore their composition (Figure~\ref{fig:composition}) is just a trace. Note the same theory applies even when the Moore machines are non-deterministic.
	In Figures~\ref{fig:contracto} and \ref{figure:contractod} we show the automaton and blame automaton for the contract $rec X.(\obl_{1}(c) \blacktriangleright \obl_{0}(b);X)$.
	Our model checking procedure (without blame) will compose Figure~\ref{fig:composition} and Figure~\ref{fig:contracto}, and identify that the trace reaches the bad state. Consider that the reparation consisting of an obligation to perform an action $b$ was not satisfied. Similarly (not shown here) blame automaton would blame party 1 for the violation. 
\end{example}

\begin{figure}[t]
	\centering
	\begin{subfigure}[htbp]{0.42 \textwidth}
		\centering
		\scalebox{0.8}{
			\begin{tikzpicture}[shorten >=1pt,node distance=3cm,on grid,auto] 
				\node[state,initial] (s_0)   {$s_0/\mathrm{a_0}$}; 
				\node[state, right=of s_0] (s_1) {$s_1/\mathrm{b_0}$}; 
				\path[->] 
				(s_0) edge[loop above] node {$a_1$} ()
				edge[bend left=10] node {$\lnot a_1$} (s_1)
				(s_1) edge[loop above] node {$\lnot b_1$} ()
				edge[bend left=10] node {$b_1$} (s_0);
		\end{tikzpicture}}
		\caption{$M_0$: Moore machine for agent 0.}
		\label{fig:moore1}
	\end{subfigure}
	\hfill
	\begin{subfigure}[htbp]{0.42 \textwidth}
		\centering
		\scalebox{0.8}{
			\begin{tikzpicture}[shorten >=1pt,node distance=3cm,on grid,auto] 
				\node[state,initial] (s_0)   {$s_0/\mathrm{c_1}$}; 
				\node[state, right=of s_0] (s_1) {$s_1/\mathrm{b_1}$}; 
				\path[->] 
				(s_0) edge[loop above] node {$b_0$} ()
				edge[bend left=10] node {$\lnot b_0$} (s_1)
				(s_1) edge[loop above] node {$ \lnot{a_0}  $} ()
				edge[bend left=10] node {$a_0$} (s_0);
		\end{tikzpicture}}
		\caption{ $M_1$: Moore machine for agent 1.}
		\label{fig:moore2}
	\end{subfigure}
	
	\vspace{5mm}
	\begin{subfigure}[htbp]{0.95 \textwidth}
		\scalebox{0.8}{
			\begin{tikzpicture}[shorten >=1pt,node distance=3cm,on grid,auto, every state/.style={inner sep=6pt,font=\scriptsize}]
				\node[state,initial] (s_0) {$s_{0},s_{0}$};
				\node[state, right=of s_0] (s_1) {$s_{1},s_{1}$};
				\node[state, right=of s_1] (s_2) {$s_{0},s_{1}$};
				\node[state, right=of s_2] (s_3) {$s_{1},s_{0}$};
				\path[->]
				(s_0) edge[] node {$\{a_0, c_1\}$} (s_1)
				(s_1) edge[] node {$\{b_0, b_1\}$} (s_2)
				(s_2) edge[] node {$\{a_0, b_1\}$} (s_3)
				(s_3) edge[loop right] node {$\{b_0, c_1\}$} ();
		\end{tikzpicture}}
		\centering{\caption{The composition $M_0 \otimes M_1$.}}
		\label{fig:composition}
	\end{subfigure}
	\vspace{5mm}
	\begin{subfigure}[b]{0.45\textwidth}
		\centering
		\scalebox{0.8}{
			\begin{tikzpicture}[shorten >=1pt,node distance=3cm,on grid,auto] 
				\node[state, initial] (q0) {$s_0$}; 
				\node[state] (q1) [right=of q0] {$s_G$}; 
				\node[state] (q2) [below=1.5cm of q0] {$s_2$};
				\node[state, accepting] (q3) [right=of q2] {$s_B$};	
				\path[->] 
				(q0) edge [bend left = 10] node {$a_0 \wedge a_1$} (q1)
				(q0) edge [] node[swap] {$\lnot (a_0, \wedge a_1)$} (q2)
				(q2) edge [right] node {$\begin{array}{l}\\b_0 \wedge b_1\end{array}$} (q1)
				(q2) edge [] node[ below] {$\lnot (b_0 \wedge b_1)$} (q3)
				(q1) edge [bend left = 10] node {$\epsilon$} (q0); 
		\end{tikzpicture}}
		\caption{$C = aut(rec X. (\obl_1(a) \blacktriangleright \obl_0(b));X)$}
		\label{fig:contracto}
	\end{subfigure}
	\hfill
	\begin{subfigure}[b]{0.45\textwidth}
		\centering
		\scalebox{0.8}{
			\begin{tikzpicture}[shorten >=1pt,node distance=3cm,on grid,auto] 
				\node[state, initial] (q0) {$s_0$}; 
				\node[state] (q1) [right=of q0] {$s_G$}; 
				\node[state] (q2) [below=2cm of q0] {$s_2$};
				\node[state, accepting] (q3) [below=1cm of q1] {$s_B^0$};
				\node[state, accepting] (q4) [below=1cm of q3] {$s_B^1$};
				\path[->] 
				(q0) edge [bend left=10] node {$a_0 \wedge a_1$} (q1)
				(q0) edge [] node[swap] {$\lnot(a_0 \wedge a_1)$} (q2)
				(q2) edge [bend left = 10] node {$\qquad b_0 \wedge b_1$} (q1)
				(q2) edge [] node[] {$\qquad \ \lnot b_0$} (q3)
				(q2) edge [] node[] {$\qquad b_0 \wedge \lnot b_1$} (q4)
				(q1) edge [bend left=10] node {$\epsilon$} (q0); 
		\end{tikzpicture}}
		\caption{$blAut(rec X. (\obl_1(a) \blacktriangleright \obl_0(b));X)$}
		\label{figure:contractod}
	\end{subfigure}
	\vspace{5mm}
	
	\begin{subfigure}[htbp]{0.95 \textwidth}
		\scalebox{0.8}{
			\begin{tikzpicture}[shorten >=1pt,node distance=1.5cm,on grid,auto, every state/.style={inner sep=6pt,font=\scriptsize}]
				\node[state,initial] (s_0) {$s_2$};
				\node[state, right=of s_0] (s_1) {$s_G$};
				\node[state, right=of s_1] (s_2) {$s_0$};
				\node[state, right=of s_2] (s_3) {$s_2$};
				\node[state, right=of s_3] (s_4) {$s_B$};
				\path[->]
				(s_0) edge node {} (s_1)
				(s_1) edge node {} (s_2)
				(s_2) edge node {} (s_3)
				(s_3) edge node {} (s_4);
				
		\end{tikzpicture}}
		\centering{\caption{$(M_0 \otimes M_1) \|C$}}
		\label{fig:moore3}
	\end{subfigure}
	\caption{Example of the model checking approach.}
	\label{fig:composition}
\end{figure}
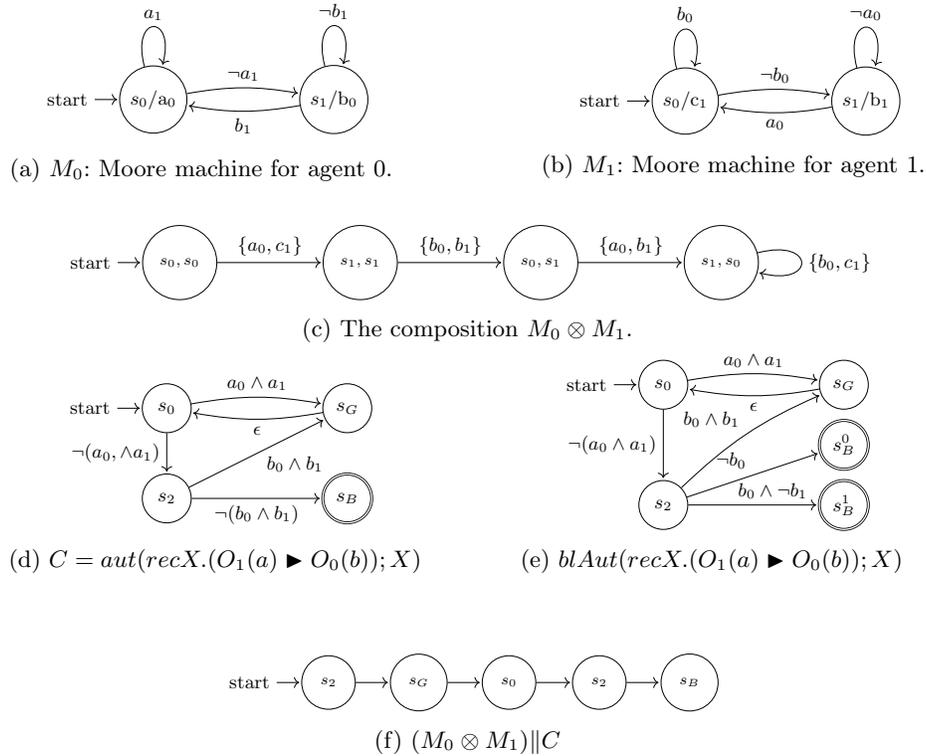

\vspace{-5pt}

\section{Related Work}\label{sec:relatedwork}

{\bf Multi-agent systems.} 
A number of logics can express properties about {multi-agent systems}. For example, ATL can express the existence of a strategy for one or more agents to enforce a certain specification \cite{alur2002alternating}, while strategy logic makes strategies first-class objects \cite{10.1007/978-3-540-74407-8_5}. Checking for the existence of strategies is in 2EXPTIME. Our logic is not concerned with the existence of strategies, but with analyzing the party strategies to ensure they respect a contract. 
So, our approach is more comparable to LTL than to game-based logic, limited to (co-)safety properties and with a notion of norms that allows us to talk about blame natively. 

Concerning blame, 
\cite{Friedenberg_Halpern_2019} considers the notion of \textit{blameworthiness}. They use structural equations to represent agents, but the approach is not temporal, and each agent performs only one action.  Work in this area (e.g., \cite{Friedenberg_Halpern_2019,10.5555/3504035.3504261,ChocklerH03}) tends to be in a different setting than ours. 

They consider the cost of actions and agents' beliefs about the probability of their actions not achieving the expected outcome. Instead, we assume all the parties have knowledge of the contract, and we take an automata-theoretic approach. Moreover, our blame derives from the norms, whereas other work depends on a notion of causality \cite{DBLP:journals/corr/Chockler16}.

The work \cite{DBLP:conf/atal/AbarcaB22} extends {\em STIT logic}  with notions of responsibility, allowing reasoning about blameworthiness and praiseworthiness. This, and other similar work (e.g., \cite{10.1093/logcom/ext072}) is more related to our work and even has a richer notion of blame. However, we give an automata-based model checking procedure.
\

\vspace*{1ex}
\noindent
{\bf Deontic logics} 
{Deontic logics} have been used in a multi-agent setting before. For example, \cite{10.1007/11786849_7} define deontic notions in terms of ATL, allowing reasoning like \textit{an obligation holds for an agent iff they have a strategy to carry it out}. These approaches (e.g., \cite{10.1007/11786849_7,10.1145/860575.860609,DBLP:conf/hybrid/Shea-BlymyerA20}) focus on obligations and neglect both reparations and our view of permissions as rights. Some approaches (e.g., \cite{10.1145/860575.860609,10.1007/978-3-540-25927-5_15}) however do perform model checking for a deontic logic in a multi-agent system setting. The work most similar to ours is that of \textit{contract automata} \cite{10.1007/s10506-016-9185-2}, wherein a contract is represented as a Kripke structure (with states tagged by norms), two parties as automata, and permissions with a similar rights-based view. However, it takes a purely operational approach, 
\
and lacks a notion of blame. 

Our language is an extension and combination of the deontic languages presented in \cite{DBLP:conf/jurix/AzzopardiGP16,DBLP:conf/jurix/AzzopardiPS18,DBLP:conf/wollic/PrisacariuS09}, combining action attempts, a right-based view of permission, a two-party setting, and regular expressions as conditions. 

Besides maintaining all these, we give denotational trace semantics, and provide blame and model checking algorithms.

\section{Conclusions}\label{sec:conclusion}

In this paper we have introduced a deontic logic
for reasoning about a two-party synchronous setting. 
This logic allows one to define constraints on when parties should support or non-interfere with the carrying out of a certain action or protocol. Using a pair of party traces, we can talk about attempts and success to perform collaborative actions. We consider automata constructions describing both the set of all satisfying and violating sequences. 
Given the behavior of the agents in the form of suitable automata, we have also provided algorithms for model checking and for blame assignment.
To differentiate between satisfying a formula in the expected manner or by fullfilling the exceptional case, we introduce a quantitative semantics. This allows ordering satisfying traces depending on how often they use these exceptions.

This work may be extended in many directions. First,  we could consider asynchronous interaction, distinguishing between sending and receiving.
\new{The syntax and semantics can also be extended easily to handle multi-party agents rather than just a two-party setting.}
Different quantitative semantics could be given, for example considering the \emph{costs of actions} to reason when it is better to pay a fine rather than to behave as expected. 
We plan to study how to synthesise strategies for the different parties, for instance to ensure the optimal behaviour of agents.

\bibliographystyle{splncs04}
\bibliography{lib}
\end{document}